\documentclass[12pt,letterpaper]{article}
\usepackage[margin=1in]{geometry}

\usepackage[utf8]{inputenc}
\usepackage[T1]{fontenc}
\usepackage{microtype}
\usepackage{amsmath}
\usepackage{amssymb}
\usepackage{amsthm}
\newtheorem{theorem}{Theorem}
\newtheorem{lemma}[theorem]{Lemma}
\newtheorem{corollary}[theorem]{Corollary}
\usepackage{mathtools}
\usepackage{tikz}
\usetikzlibrary{matrix}
\usetikzlibrary{tikzmark,calc}
\usepackage[ruled,vlined]{algorithm2e}

\usepackage{xspace}
\newcommand{\pmax}{p_{\max}}
\newcommand{\problem}{$1\,|\,|\,\sum p_j U_j$\xspace}
\newcommand{\problemweighted}{$1\,|\,|\,\sum w_j U_j$\xspace}
\newcommand{\multiplemachines}{$Pm\,|\,|\,\sum p_j U_j$\xspace}

\DeclareMathOperator*{\argmin}{arg\,min}

\let\le\leqslant
\let\leq\leqslant
\let\ge\geqslant
\let\geq\geqslant

\title{On Minimizing Tardy Processing Time, Max-Min Skewed Convolution, and Triangular Structured ILPs\thanks{Part of this research was done during the Discrete Optimization trimester program at Hausdorff Research Institute for Mathematics (HIM) in Bonn, Germany.}}
\author{Kim-Manuel Klein\thanks{Supported by the DFG project KL 3408/1-1.} \\ \normalsize{Bremen University} \and Adam Polak\thanks{Supported by the Swiss National Science Foundation project \emph{Lattice Algorithms and Integer Programming} (185030).} \\ \normalsize{EPFL} \and Lars Rohwedder \\ \normalsize{Maastricht University}}
\date{}

\begin{document}

\maketitle

\begin{abstract}
The starting point of this paper is the problem of scheduling $n$ jobs with processing times and due dates on a single machine so as to minimize the total processing time of tardy jobs, i.e., \problem. This problem was identified by Bringmann et al.~(Algorithmica 2022) as a natural subquadratic-time special case of the classic \problemweighted problem, which
likely requires time quadratic in the total processing time $P$, because of a fine-grained lower bound. Bringmann et al.~obtain their $\widetilde{O}(P^{7/4})$ time scheduling algorithm through a new variant of convolution, dubbed Max-Min Skewed Convolution, which they solve in $\widetilde{O}(n^{7/4})$ time. Our main technical contribution is a faster and simpler convolution algorithm running in $\widetilde{O}(n^{5/3})$ time. It implies an $\widetilde{O}(P^{5/3})$ time algorithm for \problem, but may also be of independent interest.

Inspired by recent developments for the Subset Sum and Knapsack problems, we study \problem parameterized by the maximum job processing time $p_{\max}$. With proximity techniques borrowed from integer linear programming (ILP), we show structural properties of the problem that, coupled with a new dynamic programming formulation, lead to an $\widetilde{O}(n+p_{\max}^3)$ time algorithm. Moreover, in the setting with multiple machines, we use similar techniques to get an $n \cdot p_{\max}^{O(m)}$ time algorithm for \multiplemachines.

Finally, we point out that the considered problems exhibit a particular triangular block structure in the constraint matrices of their ILP formulations.
In light of recent ILP research, a question that arises is
whether one can devise a generic algorithm for such a class of ILPs.
We give a negative answer to this question: we show that already a slight generalization of the structure of the scheduling ILP leads to a strongly NP-hard problem.
\end{abstract}
\pagebreak
\section{Introduction}
We consider the scheduling problem \problem of minimizing the total sum of processing times of \emph{tardy} jobs, where we are given a set of $n$ jobs, numbered from $1$ to $n$, and each job $j$ has a processing time $p_j$ and a due date $d_j$. A schedule is defined by a permutation $\sigma: \{1, \ldots ,n \} \to \{1, \ldots ,n \}$ of the jobs, and, based on this schedule $\sigma$, the \emph{completion time} of jobs is defined. The completion time $C_j$ of a job $j$ is $C_j = \sum_{i \, : \, \sigma(i) \leq \sigma(j)} p_i$. The objective of the problem is to find a schedule $\sigma$ that minimizes the sum of processing times of jobs that miss their due date $d_j$ (called \emph{tardy}), i.e.,
\[
    \min_{\sigma} \ \sum_{\mathclap{j \, : \, C_j > d_j}} p_j.
\]
Note that for tardy jobs we pay their penalty regardless of the actual completion time. Therefore, in the remainder of the paper, we will use the equivalent problem formulation that we have to select a subset of jobs $S\subseteq \{1,\ldots,n\}$ such that all selected job can be completed by their due dates.
In this case it can be assumed that the jobs from $S$ are scheduled by the earliest-due-date-first order~\cite{lawler1969functional_deadline}.
One could also think of the problem as a scenario
where the jobs that cannot be scheduled before their due dates on the available machine have to be outsourced somewhere else, and the cost of doing this is proportional to the total size of these outsourced jobs. These two properties are typical for hybrid cloud computing platforms.

For the case where all due dates are identical, the scheduling problem is equivalent to the classic Subset Sum problem. In the latter problem we are given a (multi-)set of numbers $\{a_1, \ldots , a_n\}$ and a target value $t$, and the objective is to find a subset of the numbers that sums up exactly to $t$, which is equivalent to the problem of finding the maximum subset sum that is at most $t$. With arbitrary due dates, \problem behaves like a multi-level generalization of the Subset Sum problem, where upper bounds are imposed not only on the whole sum, but also on prefix sums.

In recent years there has been considerable attention on developing fast pseudopolynomial time algorithms solving the Subset Sum problem. Most prominent is the algorithm by Bringmann~\cite{bringmann2017subset_sum_near_linear} solving the problem in time $\widetilde{O}(t)$.\footnote{The $\widetilde O$ notation hides polylogarithmic factors.}
The algorithm relies, among other techniques, on the use of Boolean convolution, which can be computed very efficiently in time $O(n \log n)$ by Fast Fourier Transform.
Pseudopolynomial time algorithms have also been studied for a parameter $a_{\max} = \max_i a_i$, which is stronger, i.e., $a_{\max} \le t$.
Eisenbrand and Weismantel~\cite{eisenbrand2018proximity} developed an algorithm for integer linear programming (ILP) that, when applied to Subset Sum, gives a running time of $O(n + a_{\max}^3)$, the first algorithm with a running time of the form $O(n + \mathrm{poly}(a_{\max}))$. Based on the Steinitz Lemma they developed proximity results, which can  be used to reduce the size of $t$ and therefore solve the problem within a running time independent of the size of the target value $t$. The currently fastest algorithm in this regard is by Polak, Rohwedder and Węgrzycki~\cite{polak_subsetsum_an} with running time $\widetilde O(n+ a_{\max}^{5/3})$.

Given the recent attention on pseudopolynomial time algorithms for Subset Sum, it is not surprising that also~\problem has been considered in this direction.
Already in the late '60s, Lawler and Moore~\cite{lawler1969functional_deadline} considered this problem or, more precisely, the general form of \problemweighted, where the penalty for each job being tardy is not necessarily $p_j$, but can be an independent weight $w_j$.
They solved this problem in time $O(nP)$, where $P$ is the sum of processing times over all jobs. Their algorithm follows from a (by now) standard dynamic programming approach.
This general form of the problem is unlikely to admit better algorithms:
under a common hardness assumption the running time of this algorithm is essentially the best possible.
Indeed, assuming that $(\min,+)$-convolution cannot be solved in subquadratic time (a common hardness assumption in fine-grained complexity), there is no algorithm that solves \problemweighted in time $O(P^{2-\epsilon})$, for any $\epsilon>0$. This follows from the fact that \problemweighted generalizes Knapsack and the hardness already holds for Knapsack~\cite{cygan2019_minconv,KunnemannPS17_finegrainedDP}.
Since the hardness does not hold for Subset Sum (the case of Knapsack where profits equal weights), one could hope that it also does not hold either for the special case of \problemweighted, where penalties equal processing times, which is precisely~\problem. Indeed, Bringmann, Fischer, Hermelin, Shabtay, and Wellnitz~\cite{bringmann_deadline_scheduling}
recently obtained a subquadratic algorithm with running time $\widetilde O(P^{7/4})$ for the problem.
Along with this main result, they also consider other parameters: the sum of distinct due dates and the number of distinct due dates.
See also Hermelin et al.~\cite{hermelin_deadline_scheduling} for more related results.

\paragraph*{Max-min skewed convolution.}
Typically, a convolution has two input vectors $a$
and $b$ and outputs a vector $c$, where $c[k] = \oplus_{i + j = k} (a[i] \otimes a[j])$.
Different kinds of operators $\oplus$ and $\otimes$ have been studied, and, if the operators require only constant time, then a quadratic time algorithm is trivial. However, if, for example, ``$\oplus$'' $=$ ``$+$'' (standard addition) and ``$\otimes$'' $=$ ``$\cdot$'' (standard multiplication), then Fast Fourier Transform can solve convolution very efficiently in time $O(n\log n)$.
If, on the other hand, ``$\oplus$'' $=$ ``$\max$'' and ``$\otimes$'' $=$ ``$+$'', it is generally believed that no algorithm can solve the problem in time $O(n^{2 - \epsilon})$ for some fixed $\epsilon > 0$~\cite{cygan2019_minconv}.
Convolution problems have been studied extensively in fine-grained complexity, partly because they serve as good subroutines for other problems, see also~\cite{BringmannKW19, Lincoln0W20} for other recent examples.
The scheduling algorithm by Bringmann et al.~\cite{bringmann_deadline_scheduling} works by first reducing the problem to a new variant of convolution, called max-min skewed convolution,
and then solving this problem in subquadratic time.
Max-min skewed convolution is defined as the problem where we are given vectors $(a[0],a[1],\dotsc,a[n-1])$ and $(b[0],b[1],\dotsc,b[n-1])$ as well as a third vector $(d[0],d[1],\dotsc,d[2n-1])$ and our goal is to compute, for each $k = 0,1,\dotsc,2n-1$, the value
\begin{equation*}
    c[k] = \max_{i + j = k}\{\min\{a[i], b[j] + d[k]\} .
\end{equation*}
This extends the standard max-min convolution, where $d[k] = 0$ for all $k$. Max-min convolution can be solved non-trivially in time $\widetilde O(n^{3/2})$~\cite{maxminconv}. Bringmann et al.~develop an algorithm with running time $\widetilde O(n^{7/4})$ for max-min skewed convolution, when $d[k] = k$ for all $k$ (which is the relevant case for \problem). By their reduction this implies an $\widetilde O(P^{7/4})$ time algorithm for~\problem.

Our first result and our main technical contribution is a faster, simpler and more general algorithm for max-min skewed convolution.
\begin{theorem}
Max-min skewed convolution can be computed
in time $O(n^{5/3} \log n)$.
\end{theorem}
As a direct consequence, we obtain an improved algorithm for~\problem.
\begin{corollary}
The problem \problem can be solved in $\widetilde O(P^{5/3})$ time, where $P$ is the sum of processing times.
\end{corollary}
Since convolution algorithms are often
used as building blocks for other problems,
we believe that this result is of independent interest, in particular, since our algorithm also works for arbitrary $d[k]$.
From a technical point of view, the algorithm
can be seen as a two-dimensional generalization of the approach used in the $\widetilde O(n^{3/2})$ time algorithm for max-min convolution~\cite{maxminconv}.
This is quite different and more compact than
the algorithm by Bringmann et al.~\cite{bringmann_deadline_scheduling}, which only
relies on the ideas in~\cite{maxminconv} indirectly by invoking max-min convolution as a black box.

\paragraph*{Parameterization by the maximum processing time.}
As mentioned before, Subset Sum has been extensively studied with respect to running times in the maximum value $a_{\max}$.
Our second contribution is that we show that running time in similar spirit can also be achieved for~\problem.
Based on techniques from integer programming, namely, Steinitz Lemma type of arguments that have also been crucial for Subset Sum, we show new structural properties of solutions for the problem. Based on this structural understanding, we develop an algorithm that leads to the following result.
\begin{theorem}
The problem~\problem can be solved in time $\widetilde O(n + p^3_{\max})$.
\end{theorem}
For the generalized problem \multiplemachines with multiple machines, we present an algorithm relying on a different structural property.
\begin{theorem}
The problem~\multiplemachines can be solved in time  $O(n \cdot \pmax^{O(m)})$.
\end{theorem}
Similar algorithmic results with running times depending on the parameter $\pmax$ have been developed for makespan scheduling and minimizing weighted completion time (without due dates) \cite{knop2017scheduling}.

\paragraph*{Integer programming generalization and lower bounds.}

The problem \problem can be formulated as an integer linear program (ILP) with binary
variables as follows:

\begin{equation}
\text{maximize} \quad \sum_j p_jx_j \quad
\text{subject to} \quad
\begin{pmatrix}
p_1 & 0 & \cdots & 0\\
p_1 & p_2 & \ddots &  \\
\vdots &  & \ddots & 0 \\
p_1 & p_2 & \cdots & p_n
\end{pmatrix}
\cdot x \le
\begin{pmatrix}d_1\\ d_2 \\ \vdots \\ d_n\end{pmatrix}, \quad x \in \{0, 1\}^n.
\label{ILP_problem}
\end{equation}
Here, variable $x_j$ indicates whether job $j$ is selected in the solution (or, in the alternative formulation, finished within its due date).
The objective is to maximize the processing time of the selected jobs.
The necessary and sufficient conditions are that
the total volume of selected jobs with due date at most some time $t$ does not
exceed $t$. It suffices to consider only constraints for $t$ equal to one of
the jobs' due dates.

Clearly, the shape of the non-zeros in the constraint matrix of the ILP exhibits a very special
structure, a certain triangular shape.
In recent years there has been significant attention in the area of structured integer programming on identifying parameters and structures for which integer programming is tractable~\cite{DBLP:conf/soda/CslovjecsekEHRW21, CslovjecsekEPVW21, JansenKL21, klein_multistage, KouteckyLO18, EisenbrandHK18}. The most prominent example in this line of work are so-called $n$-fold integer
programs (see~\cite{DBLP:conf/soda/CslovjecsekEHRW21} and references therein), which are of the form
\[
\text{maximize} \quad c^T x \quad \text{subject to} \quad
\begin{pmatrix}
A_1 & A_2 & \cdots & A_n \\
B_1 & 0 &  & 0\\
0 & B_2 & \ddots &   \\
\vdots & \ddots & \ddots &  0 \\
0 & \cdots & 0 & B_n
\end{pmatrix}
\cdot x = b \quad \text{and} \quad \forall_i \ x_i \in \mathbb{Z}_{\geqslant 0}.
\]
Here $A_i$ and $B_i$ are ``small'' matrices in the sense that it is considered
acceptable for the running time to depend superpolynomially (e.g., exponentially) on their parameters.
It is known that these types of integer programs
can be solved in FPT time $\max_i f(A_i, B_i) \cdot \mathrm{poly}(|I|)$,
where $f(A_i, B_i)$ is a function that
depends only on the matrices $A_i$ and $B_i$ (potentially superpolynomially),
but not on $n$ or any other parameters,
and $\mathrm{poly}(|I|)$ is some polynomial in the encoding length of the input~\cite{DBLP:conf/soda/CslovjecsekEHRW21}.
There exist generalizations of this result and also other tractable structures,
but none of them captures the triangular shape of~(\ref{ILP_problem}).

We now consider a natural generalization of $n$-fold ILPs that also contains
the triangle structure~(\ref{ILP_problem}), in the remainder referred to as \emph{triangle-fold} ILPs.
\[
\text{maximize} \quad c^T x \quad \text{subject to} \quad
\begin{pmatrix}
A_1 & 0 & \cdots & 0\\
A_1 & A_2 & \ddots &  \\
\vdots &  & \ddots & 0 \\
A_1 & A_2 & \cdots & A_n \\
B_1 & 0 &  & 0\\
0 & B_2 & \ddots &   \\
\vdots & \ddots & \ddots &  0 \\
0 & \cdots & 0 & B_n
\end{pmatrix}
\cdot x \le b \quad \text{and} \quad \forall_i \ x_i \in \mathbb{Z}_{\geqslant 0}.
\]

Let us elaborate on some of the design choices and why they come naturally.
First, it is obvious that this structure generalizes $n$-fold, because
constraints can be ``disabled'' by selecting a right-hand side of $\infty$
(or some sufficiently large number). After disabling a prefix of constraints, we
end up with exactly the $n$-fold structure except that we have inequality ($\le$)
constraints instead of equality constraints. Clearly, equality constraints can
easily be emulated by duplicating and negating the constraints.
The reason we chose inequality is that with equality constraints this ILP would
directly decompose into independent small subproblems, which would be uninteresting from an algorithmic point of view and also not general enough to capture~(\ref{ILP_problem}).
The matrices $B_1, \dotsc, B_n$ can, for example, be used to express lower and upper bounds on variables, such as the constraints $x_i \in \{0, 1\}$ in~(\ref{ILP_problem}).

With this formulation it is also notable that the
scheduling problem on multiple machines, \multiplemachines, can be modelled easily: instead of one decision variable $x_j \in \{0, 1\}$ for each job $j$, we introduce
variables $x_{j,1},x_{j,2},\dotsc,x_{j,m} \in \{0, 1\}$, one for each machine.
Then we use the constraints $p_1 x_{1,i} + p_2 x_{2,i} + \cdots + p_j x_{j,i} \le d_j$
for each machine $i$ and job $j$, which ensure that on each machine
the selected jobs can be finished within their respective due date.
Finally, we use constraints of the form $x_{j,1} + x_{j,2} + \cdots + x_{j,m} \le 1$ to guarantee that each job is scheduled on at most one machine.
It can easily be verified that these constraints have the form of a triangle-fold
where the small submatrices have dimension only dependent on $m$.

Given the positive results for \problem and \multiplemachines, one may hope
to develop a general theory for triangle-folds. Instead of specific techniques
for these (and potentially other) problems, in this way one could create general
techniques that apply to many problems. For example, having an algorithm with the running time of
the form $\max_i f(A_i, B_i) \cdot \mathrm{poly}(|I|)$ (like we have for $n$-folds), one would directly get an FPT algorithm for \multiplemachines with parameters $p_{\max}$ and $m$.
However, we show strong hardness results, which indicate that
such a generalization is not possible.
\begin{theorem}\label{th:trianglefold}
  There exist fixed matrices $A_i, B_i$ of constant dimensions for which
  testing feasibility of triangle-fold ILPs is NP-hard.
  This holds even in the case when $A_1 = A_2 = \cdots = A_n$ and
  $B_1 = B_2 = \cdots = B_n$.
\end{theorem}
This hardness result has a surprising quality: not only are
there no FPT algorithms for triangle-folds, but also no XP algorithms, which
is in stark contrast to other similarly shaped cases.

\section{Max-min skewed convolution}
Recall that max-min skewed convolution is defined as the problem where, given
vectors $a$, $b$ and $d$, we want to compute, for each $k = 0,1,\dotsc,2n-1$, the value
\begin{equation*}
    c[k] = \max_{i + j = k}\{\min\{a[i], b[j] + d[k]\} .
\end{equation*}
In this section we give an algorithm that solves this problem in time $O(n^{5/3}\log n)$.
\begin{figure}
    \centering
    \begin{align*}
    a &= (5, 6, 4, 9, 1, 7, 3, 8, 2) \\
    b &= (8, 4, 1, 2, 3, 7, 6, 5, 9) 
    \end{align*}
    \begin{tabular}{c c c c c c c c c c}
    & $b[0]$ & $b[1]$ & $b[2]$ & $b[3]$ & $b[4]$ & $b[5]$ & $b[6]$ & $b[7]$ & $b[8]$ \\

$a[0]$ & 0 & 1 & 1 & 1 & 1 & 1 & 1 & 1 & 0 \\
$a[1]$ & 0 & 1 & 1 & 1 & 1 & 1 & 1 & 1 & 0 \\
$a[2]$ & 0 & 1 & 1 & 1 & 1 & 1 & 1 & 1 & 0 \\
$a[3]$ & 0 & 0 & 1 & 1 & 0 & 0 & 0 & 0 & 0 \\
$a[4]$ & 1 & 1 & 1 & 1 & 1 & 1 & 1 & 1 & 0 \\
$a[5]$ & 0 & 1 & 1 & 1 & 1 & 0 & 0 & 0 & 0 \\
$a[6]$ & 1 & 1 & 1 & 1 & 1 & 1 & 1 & 1 & 0 \\
$a[7]$ & 0 & 0 & 1 & 1 & 0 & 0 & 0 & 0 & 0 \\
$a[8]$ & 1 & 1 & 1 & 1 & 1 & 1 & 1 & 1 & 0 \\
    \end{tabular}
    \caption{Matrix $M_k$ for $k = 6$ and vectors $a$ and $b$, which it is derived from.
    The entry in cell for row $a[i]$ and column $b[j]$ describes whether there exist $i', j'$ with $i' + j' = k$, $a[i'] \ge a[i]$ and $b[j'] \ge b[j]$. For example, the entry for $a[2]$ and $b[2]$ is $1$ because $a[3] \geq a[2]$ and $b[3] \geq b[2]$.}
    \label{fig:conv-preprocessing}
\end{figure}
Assume without loss of generality that the values in $a$ and $b$
are all different. This can easily be achieved by multiplying all values (including those in $d$) by $4 n$ and adding a different number from $0,1,\dotsc,2n-1$ to each entry in $a$ and $b$.
After computing the convolution we then
divide the output numbers by $4n$ and round down to
reverse the effect of these changes.

The algorithm consists of two phases. In the first phase we build a data
structure to aid the computation of each $c[k]$. Then in the second phase we compute each element $c[k]$ separately with a binary search.
The goal of the data structure is thus to efficiently answer queries
of the form ``is $c[k] > v$?'' for some given $k$ and $v$.

First we introduce a quite excessive data structure.
Building it explicitly is impossible within our running time goals, but it will help gain intuition. Suppose for every $k$ we have a matrix $M_k$ (see Figure~\ref{fig:conv-preprocessing}), where the rows correspond to the values $a[0],a[1],\dotsc,a[n-1]$, the columns
correspond to $b[0],b[1],\dotsc,b[n-1]$, and the entries tell us whether, for some $a[i], b[j]$, there exist $i'$ and $j'$ with $i' + j' = k$, $a[i'] \ge a[i]$ and $b[j'] \ge b[j]$. This matrix contains
enough information to answer the queries above:
To determine whether $c[k] > v$, we need to check if
there are $i', j'$ with $i' + j' = k$ and $a[i'] > v$ and $b[j'] > v - d[k]$.
Choose the smallest $a[i]$ such that $a[i] > v$ and the smallest $b[j]$ such that
$b[j] > v - d[k]$. If such $a[i]$ or $b[j]$ does not exist, then we already know that $c[k] \le v$. Otherwise, if both elements exist, the entry $M_k[a[i], b[j]]$ directly gives us the answer:
If it is $0$, then for all $i' + j' = k$ either (1) $a[i'] < a[i]$ and hence (since $a[i]$ is the smallest element greater than $v$) $a[i'] \le v$, or (2) $b[j'] < b[j]$ and hence $b[j'] \le v - d[k]$; thus $c[k] \le v$. The converse is also true: Suppose that $c[k] \le v$. Then for all $i'$, $j'$ with $i' + j' = k$
we have that either $a[i'] \le v < a[i]$ or $b[j'] \le v - d[k] < b[j]$, and thus the matrix entry at $(a[i], b[j])$ is $0$.

It is obvious that we cannot afford to explicitly construct the matrices described above; their space alone would be cubic in $n$. Instead, we are only going to compute a few carefully chosen values of these matrices, that allow us to recover any other value in sublinear time.
First, reorder the columns and rows so that
the corresponding values ($a[i]$ or $b[j]$) are increasing.
We call the resulting matrix $M^{\mathrm{sort}}_k$, see also Figure~\ref{fig:conv-preprocessing2}.
Clearly, this does not change the information stored in it, but one can observe that now the rows and columns are each nonincreasing.
We will compute only the values at intersections of every $\lfloor n/p \rfloor$-th row and every $\lfloor n / p \rfloor$-th column, for a parameter $p \in \mathbb N$, which is going to be specified later. This means we are computing a total of $O(n p^2)$ many values. Although computing a single entry of one of the matrices $M^{\mathrm{sort}}_k$ would require linear time, we will show that computing the same entry for all matrices (all $k$) can be done much more efficiently than in $O(n^2)$ time.

Indeed, for fixed $u, w \in \mathbb Z$, it takes only $O(n \log n)$ time to compute, for all $k$, whether there exists $i'$ and $j'$ with $i' + j' = k$,  $a[i'] \ge u$ and $b[j'] \ge v$. This fact follows from a standard application of Fast Fourier Transformation (FFT): we construct two vectors $a', b'$ where
\begin{equation*}
    a'[i] = \begin{cases}
    1 &\text{ if } a[i] \ge u, \\
    0 &\text{ otherwise,}
    \end{cases}
    \quad
    b'[i] = \begin{cases}
    1 &\text{ if } b[j] \ge w, \\
    0 &\text{ otherwise,}
    \end{cases}
\end{equation*}
and compute their $(+, \cdot)$-convolution with FFT. For non-zero output entries there exist $i', j'$ as above, and for zero entries they do not.
It follows that computing the selected $O(n p^2)$ values of
$M^{\mathrm{sort}}_0, M^{\mathrm{sort}}_1, \ldots, M^{\mathrm{sort}}_{2n-2}$
can be done in time $O(p^2 \cdot n \log n)$. 
\begin{figure}
    \centering
    \begin{tabular}{c c c c c c c c c c}
       & $b[2]$ & $b[3]$ & $b[4]$ & $b[1]$ & $b[7]$ & $b[6]$ & $b[5]$ & $b[0]$ & $b[8]$ \\
$a[4]$ & 1 & 1 & 1 & 1 & 1 & 1 & 1 & 1 & 0 \\
$a[8]$ & 1 & 1 & 1 & 1 & 1 & 1 & 1 & 1 & 0 \\
$a[6]$ & 1 & 1 & \fbox{1} & 1 & 1 & \fbox{1} & 1 & 1 & \fbox{0} \\
$a[2]$ & 1 & 1 & 1 & 1 & 1 & 1 & 1 & 0 & 0 \\
$a[0]$ & 1 & 1 & 1 & 1 & 1 & 1 & 1 & 0 & 0 \\
$a[1]$ & 1 & 1 & \fbox{1} & 1 & 1 & \fbox{1} & 1 & 0 & \fbox{0} \\
$a[5]$ & 1 & 1 & 1 & 1 & 0 & 0 & 0 & 0 & 0 \\
$a[7]$ & 1 & 1 & 0 & 0 & 0 & 0 & 0 & 0 & 0 \\
$a[3]$ & 1 & 1 & \fbox{0} & 0 & 0 & \fbox{0} & 0 & 0 & \fbox{0} \\
    \end{tabular}

    \caption{Matrix $M^{\mathrm{sort}}_k$, which is identical to that in Figure~\ref{fig:conv-preprocessing} except for reordering of rows and columns. Highlighted are the elements that we compute during preprocessing (assuming $p=3$).}
    \label{fig:conv-preprocessing2}
\end{figure}

We now consider the second phase, where the algorithm computes $c[k]$, for each $k$ separately, using binary search. To this end, we design a procedure (see Algorithm~\ref{alg:bs}) to determine, given $k \in \{0,1,2,\ldots,2n-2\}$ and $v \in \mathbb{Z}$, whether $c[k] > v$. The procedure will run in $O(n/p)$ time. As described in the beginning of the proof, this corresponds to computing the value of a specific cell $(a[i], b[j])$ in the matrix $M^{\mathrm{sort}}_k$. If this cell happens to be among the precomputed values, we are done.
Otherwise, consider the $\lfloor n/p \rfloor \times \lfloor n/p \rfloor$ submatrix that
encloses $(a[i], b[j])$ and whose corners are among the precomputed values.
If the lower right corner $(a[i''], b[j''])$ is equal to one,
then entry $(a[i], b[j])$ must also be one by monotonicity.
Hence, assume otherwise. The entry $(a[i], b[j])$ could
still be one, but this happens only if
there is a \emph{witness} $(i', j')$ that satisfies $i' + j' = k$ and
\begin{enumerate}
    \item $a[i''] > a[i'] \ge a[i]$ and $b[j'] \ge b[j]$, or
    \item $b[j''] > b[j'] \ge b[j]$ and $a[i'] \ge a[i]$.
\end{enumerate}
The number of possible witnesses for the first case is bounded by $n / p$, since there are only $\lfloor n/p \rfloor$ many values $a[i']$ between $a[i]$ and $a[i'']$ (since they are in the same $\lfloor n/p \rfloor \times \lfloor n/p \rfloor$ submatrix) and the corresponding $j'$ is fully determined by $i'$. Likewise, there are at most $n / p$ many possible witnesses for the second case. Hence, we can compute the value of the cell $(a[i], b[j])$ by exhaustively checking all these candidates for a witness, i.e.,
\[\{ (i', k-i') \mid a[i'] \in [a[i], a[i'']) \} \cup \{ (k-j', j') \mid b[j'] \in [b[j], b[j'']) \}.\]

\begin{algorithm}
\caption{Procedure used in binary search to check whether $c[k] > v$.}
\label{alg:bs}
$\mathrlap{i}\hphantom{j} \longleftarrow \argmin \{a[i] \mid a[i] > v\}$\;
$j \longleftarrow \argmin \{b[j] \mid b[j] > v - d[k]\}$\;
\lIf{$i$ \rm{or} $j$ \rm{does not exist}}{\KwRet{\textsc{no}}}
$(a[i''], b[j'']) \longleftarrow$ the closest precomputed cell below and to the right of $M^{\mathrm{sort}}_k[a[i], b[j]]$\;
\lIf{$M^{\mathrm{sort}}_k[a[i''],b[j'']] = 1$}{\KwRet{\textsc{yes}}}
\ForEach{$i' \in \{i' \mid a[i] \le a[i'] < a[i'']\}$}{
  \lIf{$b[k - i'] \ge b[j]$}{\KwRet{\textsc{yes}}}}
\ForEach{$j' \in \{j' \mid b[j] \le b[j'] < b[j'']\}$}{
  \lIf{$a[k - j'] \ge a[i]$}{\KwRet{\textsc{yes}}}}
\KwRet{\textsc{no}}\;
\end{algorithm}

Finally, let us note that, for a fixed $k$, the value of $c[k]$ has to be among the following $2n$ values: $a[0], a[1], \ldots, a[n-1], b[0] + d[k], b[1] + d[k], \ldots, b[n-1] + d[k]$. A careful binary search only over these values makes the number of iterations logarithmic in $n$, and not in the maximum possible output value. Indeed, after the preprocessing, we already have access to sorted versions of the lists $a[0],a[1],\dotsc,a[n-1]$ and $b[0],b[1],\dotsc,b[n-1]$ and, in particular, to a sorted version of $b[0] + d[k], b[1] + d[k], \ldots, b[n-1] + d[k]$ (since this is just a constant offset of the latter list). We then first binary search for the lowest upper bound of $c[k]$ in $a[0],a[1],\dotsc,a[n-1]$ and then in $b[0] + d[k], b[1] + d[k], \ldots, b[n-1] + d[k]$ in order to determine the exact value of $c[k]$.

The total running time of both phases is $O(p^2 \cdot n \log(n) + n \cdot \log(n) \cdot n/p)$. We set $p = n^{1/3}$ in order to balance the two terms, and this gives the desired $O(n^{5/3} \log n)$ running time.

\section{Parameterizing by the maximum processing time}
In this section we study algorithms for~\problem with running time optimized for $p_{\max}$ and $n$ instead of $P$.
We present an algorithm with a running time of $\widetilde O(n + p_{\max}^3)$.
Such a running time is particularly appealing when $n$ is much larger than $p_{\max}$. This complements Lawler and Moore's algorithm with complexity $O(nP) \le O(n^2 p_{\max})$, which
is fast when $n$ is small.
Interestingly, in both cases we have roughly cubic dependence on $n + p_{\max}$.

Our result is based on a central structural property that we prove for an optimal solution and
a sophisticated dynamic programming algorithm that recovers
solutions of this form. The structural property uses
exchange arguments that are similar to an approach used
by Eisenbrand and Weismantel~\cite{eisenbrand2018proximity}
to prove proximity results in integer programming via
the Steinitz Lemma.

In the following,
a solution is characterized by the subset of jobs
that are finished within their due date. Once such
a subset is selected, jobs in this subset can be scheduled in non-decreasing order sorted by their due date.
In the remainder we assume that $d_1 \le d_2 \le \cdots \le d_n$.
This sorting can be done efficiently:
we may assume without loss of generality that
all due dates are at most $np_{\max}$.
Then radix sort requires only time $O(n \log_n (n p_{\max})) \le O(n + p_{\max})$.
\begin{lemma}\label{lem:pmax-structur}
  There exists an optimal solution $S\subseteq \{1,2,\dotsc,n\}$
  such that for all $i=1,2,\dotsc,n$ it holds that either
  \begin{enumerate}
      \item $|\{1,2,\dotsc,i\} \cap S| < 2 p_{\max}$, or
      \item $|\{i + 1, i + 2,\dotsc, n\} \setminus S| < 2 p_{\max}$.
  \end{enumerate}
\end{lemma}
\begin{proof}
Let $S$ be an optimal solution that does not satisfy this property for some $i$.
It is easy to see that if we find some $A \subseteq S \cap \{1,2,\dotsc,i\}$ and $B \subseteq \{i+1,i+2,\dotsc,n\} \setminus S$ with the same volumes (that is, $\sum_{j\in A} p_j = \sum_{j\in B} p_j$), then
$(S \setminus A) \cup B$ would form an optimal solution that is closer to satisfying the property. Note that the solution $(S \setminus A) \cup B$ is feasible since the due dates of the jobs in $B$ are strictly larger than the due dates of the jobs in $A$.

Since both 1.~and 2.~are false for $i$, we have that $A' = S \cap \{1,2,\dotsc,i\}$ and $B' = \{i+1, i+2,\dotsc,n\} \setminus S$ both have cardinality at least $2 p_{\max}$.
We construct $A$ and $B$ algorithmically as follows. Starting with $A_1 = B_1 = \emptyset$
we iterate over $k = 1,2,\dotsc, 2 p_{max}$.
In each iteration we
check whether $\sum_{j\in A_k} p_j - \sum_{j\in B_k} p_j$ is positive or not. If it is positive, we set
$B_{k+1} = B_k \cup \{j\}$ for some $j\in B'\setminus B_k$ and $A_{k+1} = A_k$;
otherwise we set $A_{k+1} = A_k \cup \{j\}$ for some $j \in A'\setminus A_k$.
The difference $\sum_{j\in A_k} p_j - \sum_{j\in B_k} p_j$ is always between $-p_{\max}$
and $p_{\max} - 1$. Hence, by pidgeon-hole principle there are two indices $k < h$ such that
$\sum_{j\in A_k} p_j - \sum_{j\in B_k} p_j = \sum_{j\in A_h} p_j - \sum_{j\in B_h} p_j$.
Since $A_k \subseteq A_h$ and $B_k \subseteq B_h$ by construction, we can simply set
$A = A_h \setminus A_k$ and $B = B_h \setminus B_k$ and it follows that $\sum_{j\in A} p_j = \sum_{j\in B} p_j$.
\end{proof}

\begin{corollary}\label{cor:pmax-structur}
  There exists an optimal solution $S\subseteq \{1,2,\dotsc,n\}$ and an index $i\in\{1,2,\dotsc,n\}$
  such that
  \begin{enumerate}
      \item $|\{1,2,\dotsc,i\} \cap S| \le 2 p_{\max}$, and
      \item $|\{i + 1, i + 2,\dotsc, n\} \setminus S| < 2 p_{\max}$.
  \end{enumerate}
\end{corollary}
\begin{proof}
Consider the solution $S$ as it Lemma~\ref{lem:pmax-structur} and let $i$ be the maximum index
such that $|S \cap \{1,2,\dotsc,i\}| < 2 p_{\max}$.
If $i = n$ the corollary's statement follows. Otherwise, we set $i' = i+1$. Then
$|S \cap \{1,2,\dotsc,i'\}| = 2 p_{\max}$ and by virtue of Lemma~\ref{lem:pmax-structur} it must hold that
$|S \setminus \{i' + 1, i' + 2,\dotsc, n\}| < 2 p_{\max}$.
\end{proof}

Although we do not know the index $i$ from Corollary~\ref{cor:pmax-structur},
we can compute efficiently an index, which is equally good for our purposes.
\begin{lemma}\label{lem:pmax-idx}
In time $O(n)$ we can compute an index $\ell$ such that there exists
an optimal solution $S$ with
\begin{enumerate}
    \item $p(\{1,2,\dotsc,\ell\} \cap S) \le O(p^2_{\max})$, and
    \item $p(\{\ell+1,\ell+2,\dotsc,n\} \setminus S) \le O(p^2_{\max})$,
\end{enumerate}
where $p(X) = \sum_{i \in X} p_i$, for a subset $X \subseteq \{1, \ldots , n \}$.
\end{lemma}
\begin{proof}
For $j=1,2,\dotsc,n$ let $t_j$ denote the (possibly negative)
maximum time such that we can schedule all jobs $j,j+1,\dotsc,n$
after $t_j$ and before their respective due dates.
It is easy to see that
\begin{equation*}
    t_j = \min_{j' \ge j} \bigg( d_{j'} - \sum_{j \le k \le j'} p_k\bigg) \ .
\end{equation*}
It follows that $t_{j} = \min\{t_{j+1}, d_j\} - p_{j}$ and thus
we can compute all values $t_j$ in time $O(n)$.
Now let $h$ be the biggest index such that $t_h < 0$.
If such an index does not exist, it is trivial to compute the optimal
solution by scheduling all jobs.
Further, let $k < h$ be the biggest index such that
$\sum_{j=k}^h p_j > 2 p_{\max}^2$ or $k = 1$ if no such index exists.
Similarly, let $\ell > h$ be the smallest index such that
$\sum_{j=h}^{\ell} p_j > 4 p_{\max}^2$ or $\ell = n$ if this does not exist.
Let $i$ be the index as in Corollary~\ref{cor:pmax-structur}.
We will now argue that either $k \le i \le \ell$
or $\ell$ satisfies the claim trivially.
This finishes the proof, since 
$p(\{i, i+1, \dotsc, \ell\}) \le p(\{k, k+1, \dotsc, \ell\}) \le O(p_{\max}^2)$ and thus the properties in this lemma, which $i$ satisfies due to Corollary~\ref{cor:pmax-structur}, 
transfer to $\ell$.

As for $k\le i$, we may assume that $k > 1$
and therefore $\sum_{j=k}^{\ell} p_j > 2 p^2_{\max}$.
Notice that there must jobs in $\{k,k+1,\dotsc,n\}$
of total volume more than $2 p_{\max}^2$, which are not in $S$.
This is because $t_k < t_h - 2 p^2_{\max} < - 2 p^2_{\max}$.
On the other hand, from Property~1 of Corollary~\ref{cor:pmax-structur}
it follows that $p(\{i+1,i+2,\dotsc,n\}\setminus S) < 2 p_{\max}^2$.
This implies that $k\le i$.

We will now show that $i \le \ell$ or $\ell$ satisfies the lemma's statement trivially.
Assume w.l.o.g.\ 
that $\ell < n$ and thus $\sum_{j=h}^{\ell} p_j > 4 p^2_{\max}$. Notice that $t_\ell \ge t_h + 4 p_{\max}^2 \ge 2 p_{\max}^2 + p_{\max}$. 
If $S \cap \{1,2,\dotsc,\ell\}$ contains jobs of a total processing time more than $2 p_{\max}^2$,
then $i\le \ell$ follows directly because of Corollary~\ref{cor:pmax-structur}. Conversely, if the total
processing time is at most $2 p_{\max}^2$, then we can schedule all these jobs and also all jobs in $\{\ell + 1, \ell + 2, \dotsc, n\}$ (since $t_{\ell} \ge 2 p_{\max}^2$), which must
be optimal.
Hence, $\ell$ satisfies the properties of the lemma.
\end{proof}
An immediate consequence of the previous lemma is that
we can estimate the optimum up to an additive error of $O(p_{\max}^2)$.
We use this to perform a binary search with $O(\log(p_{\max}))$
iterations. In each iteration we need to check if there
is a solution of at least a given value $v$. For this it suffices
to devise an algorithm that decides whether there exists a solution
that runs a subset of jobs without any idle time (between $0$ and $d_n$)
or not. To reduce to this problem, we add
dummy jobs with due date $d_n$ and total processing time $d_n - v$;
more precisely, $\min \{p_{\max}, d_n - v\}$ many jobs with processing time $1$ and $\max\{0, \lceil (d_n - v - p_{\max}) / p_{\max} \rceil\}$
many jobs with processing time $p_{\max}$.
Using that w.l.o.g.\ $d_n \le n p_{\max}$, we can bound the total
number of added jobs by $O(n + p_{\max})$, which is insignificant
for our target running time.
If there exists a schedule without any idle time, we remove the dummy jobs and obtain a schedule where jobs with total processing time at least $d_n - (d_n - v) = v$ finish within their due dates. If on the other hand there is a schedule for a total processing time $v' \ge v$, we can add dummy jobs of size exactly $d_n - v'$ to arrive at
a schedule without idle time.
For the remainder of the section we consider this decision problem.
We will create two data structures that
allow us to efficiently query information about the two partial solutions from Lemma~\ref{lem:pmax-idx}, that is, the solution for jobs $1,2,\dotsc,\ell$ and that for $\ell+1,\ell+2,\dotsc,n$.

\begin{lemma}\label{lem:pmax-dyn1}
  In time $O(n + p_{\max}^3 \log(p_{\max}))$ we can compute
  for each $T = 1,2,\dotsc, O(p_{\max}^2)$ whether there is a subset
  of $\{1,2,\dotsc,\ell\}$ which can be run exactly during $[0, T]$.
\end{lemma}
\begin{proof}
We will efficiently compute for each $T = 0,1,\dotsc,O(p_{\max}^2)$
the smallest index $i_T$
such that there exists some $A_T\subseteq \{1,2,\dotsc,i_T\}$, which runs on the machine exactly for the time interval $[0,T]$. Then it only remains
to check if $i_T \le \ell$ for each $T$.

Consider $i_T$ for some fixed $T$. Then $i_T\in A_T$, since otherwise
$i_T$ would not be minimal. Further, we can assume that job $i_T$ is processed exactly during $[T - p_{i_T}, T]$, since $i_T$
has the highest due date among jobs in $A_T$. It follows that $i_{T'} < i_T$ for $T' = T - p_i$.
Using this idea we can calculate each $i_T$ using dynamic programming. Assume that we have already computed $i_{T'}$ for
all $T' < T$. Then to compute $i_T$ we iterate over all processing times $k=1,2,\dotsc,p_{\max}$.
We find the smallest index $i$ such that $p_i = k$, $i_{T - k} < i$, and $d_i \ge T$. Among the indices we get for each $k$ we
set $i_T$ as the smallest, that is,
\[i_T = \min_{k \in [p_{\max}]} \min \, \{ i \mid p_i = k, i > i_{T-k}, d_i \ge T\} . \]
The inner minimum can be computed using a binary
search over the $O(p_{\max}^2)$ jobs with $p_i = k$ and
smallest due date, since the values of $i_T$ do not
change when restricting to the $O(p_{\max}^2) \ge T$ jobs with
smallest due date.
It follows that
computing $i_T$ can be done in $O(p_{\max} \cdot \log(p_{\max}))$.
Computing $i_T$ for all $T$ requires
time $O(p_{\max}^3 \cdot \log(p_{\max}))$.
\end{proof}

\begin{lemma}\label{lem:pmax-dyn2}
  In time $O(n + p_{\max}^3 \log(p_{\max}))$ we can compute
  for each $T = 1,2,\dotsc, O(p_{\max}^2)$ whether there is some
  $A\subseteq\{\ell+1,\ell+2,\dotsc,n\}$,
  such that $\{\ell+1,\ell+2,\dotsc,n\} \setminus A$
  can be run exactly during $[d_n + T - \sum_{j=\ell+1}^n p_j, d_n]$.
\end{lemma}
\begin{proof}
For every $T = 1,\dotsc,O(p_{\max}^2)$ determine $i_T$,
which is the largest value such that there exists a
set $A_T \subseteq \{i_T,i_T + 1,\dotsc,n\}$ with
$\sum_{j\in A_T} p_j = T$, such that
$\{i+1,i+2,\dotsc,n\} \setminus A_T$
can be run exactly during the interval $[d_n + T - \sum_{j=i+1}^n p_j, d_n]$.
Consider one such $i_T$ and corresponding $A_T$.
Let $j\in A_T$ has the minimal due date. Then for $T' = T - p_j$
we have $i_{T'} > i_T$.
To compute $i_T$ we can proceed as follows. Guess the processing
time $k$ of the job in $A_T$ with the smallest due date.
Then find the maximum $j$ with $p_j = k$ and $j < i_{T - k}$.
Finally, verify that there is indeed a schedule.
For this consider the schedule of all jobs
$\{\ell+1,\ell+2,\dotsc,n\}$, where we run them as late as possible
(a schedule that is not idle in $[d_n - \sum_{i=\ell+1}^n p_i, d_n]$.
Here, we look at the ``earliness'' of each job in $\{\ell+1,\dotsc,j-1\}$.
This needs to be at least $T$ for each of the jobs.

We can check the earliness efficiently by precomputing the mentioned schedule and building a Cartesian tree with the earliness values.
This data structure can be built in $O(n)$ and allows us to query
for the minimum value of an interval in constant time.
\end{proof}
We can now combine Lemmas~\ref{lem:pmax-idx}-\ref{lem:pmax-dyn2} to
conclude the algorithm's description.
Suppose that $S$ is a set of jobs corresponding to a solution,
where the machine is busy for all of $[0, d_n]$.
By Lemma~\ref{lem:pmax-idx} there exist some $T, T' \le O(p_{\max}^2)$
such that $p(\{1,2,\dotsc,\ell\} \cap S) = T$
and $p(\{\ell+1,\ell+2,\dotsc,n\} \setminus S) = T'$.
In particular, the machine will be busy with jobs of $\{1,2,\dotsc,\ell\} \cap S$ during $[0, T]$ and
with jobs of $\{\ell + 1, \ell + 2,\dotsc,n\} \cap S$ during $[T, d_n]$.
The choice of $T$ and $T'$ implies that
\begin{equation*}
    d_n = p(S) = p(\{1,\dotsc,\ell\} \cap S) + p(\{\ell + 1,\dotsc,n\} \cap S) = T + p(\{\ell + 1, \dotsc, n\}) - T' .
\end{equation*}
Hence, $T = d_n + T' - p(\{\ell + 1, \dotsc, n\})$. In order to find
a schedule where the machine is busy between $[0, d_n]$ we proceed
as follows.
We iterate over every potential $T = 0,1,\dotsc,O(p_{\max}^2)$.
Then we check using Lemma~\ref{lem:pmax-dyn1}
whether there exists a schedule of jobs in $\{1,2,\dotsc,\ell\}$
where the machine is busy during $[0, T]$. For the correct choice
of $T$ this is satisfied. Finally, using Lemma~\ref{lem:pmax-dyn2}
we check whether there is a subset of jobs in $\{\ell+1,\ell+2, \dotsc, n\}$ that can be run during $[T, d_n] = [d_n + T' - p(\{\ell + 1, \dotsc, n\}), d_n]$.
The preprocessing of Lemmas~\ref{lem:pmax-idx}-\ref{lem:pmax-dyn2}
requires time $O(n + p^3_{\max} \log(p_{\max}))$. Then
determining the solution requires only $O(p_{\max}^2)$,
which is dominated by the preprocessing. Finally,
notice that we lose another factor of $\log(p_{\max})$ due
to the binary search.

\section{Multiple machines}
In this section we present an algorithm that
solves~\multiplemachines in $n \cdot p^{O(m)}_{\max}$ time.
We assume that jobs are ordered non-decreasingly by due dates (using radix sort as in the previous section), and that all considered schedules use earliest-due-date-first on each machine.
We will now prove a structural lemma that enables our result.
\begin{lemma}\label{lem:Pm}
  There is an optimal schedule such that
  for every job $j$ and every pair of machines $i, i'$ we have
  \begin{equation}\label{eq:machines-balance}
      |\ell_i(j) - \ell_{i'}(j)| \le O(p_{\max}^2) ,
  \end{equation}
  where $\ell_i(j)$ denotes the total volume of jobs $1,2,\dotsc,j$ scheduled on machine $i$.
  Furthermore, the schedule satisfies, for every $j$ and $i$, that
  \begin{equation}\label{eq:machines-lb}
      \ell_i(j) \le \min_{j' > j}\left\{ d_{j'} - \frac{1}{m} \sum_{j''=j+1}^{j'} p_{j''} \right\} + O(p_{\max}^2) .
  \end{equation}
\end{lemma}
\begin{proof}
  The proof relies on swapping arguments.
  As a first step we make sure that between $\ell_1(n),\dotsc,\ell_m(n)$
  the imbalance is at most $p_{\max}$. If it is bigger, we simply move the
  last job from the higher loaded machine to the lower loaded machine.
  
  Now we augment this solution to obtain that, for every two machines $i$, $i'$ and every~$j$, we satisfy~\eqref{eq:machines-balance}.
  Let $j$, $i$, and $i'$ be such that $\ell_i(j) > \ell_{i'}(j) + 3 p_{\max}^2$.
  This implies that on $i'$ there is a set of jobs $A$
  with due dates greater than $d_j$ that is processed between $\ell_{i'}(d_j)$
  and at least $\ell_i(d_j) - p_{\max}$ (the latter because of the balance of total machine loads). Thus, $p(A) \ge 3 p^2_{\max} - p_{\max}$ and consequently $|A| \ge 2 p_{\max}$.
  On $i$ let $B$ be the set of jobs with due dates at most 
  $d_j$ that are fully processed
  between $\ell_{i'}(d_j)$ and $\ell_i(d_j)$. Then also
  $p(B) \ge 3 p^2_{\max} - p_{\max}$ and $|B| \ge 2 p_{\max}$.
  By the same arguments
  as in Lemma~\ref{lem:pmax-structur}, there are non-empty
  subsets $A'\subseteq A$ and $B' \subseteq B$ with equal processing time.
  We swap $A'$ and $B'$, which improves the balance between $\ell_i(j)$
  and $\ell_{i'}(j)$.
  We now need to carefully apply these swaps so that after a finite number of them we have no large imbalances anymore.
  We start with low values of $j$, balance all $\ell_i(j)$, and then increase $j$.
  However, it might be that balancing
  $\ell_i(j)$ and $\ell_{i'}(j)$ affects the values $\ell_i(j')$ for $j' < j$. To avoid this, we will use some buffer space. Notice that because
  \begin{equation*}
      \sum_i \ell_i(j) = \sum_{j' : \text{$j'$ is scheduled and } j' \le j} p_{j'} ,
  \end{equation*}
  it follows that a swap does not change the average $\ell_i(j)$ over all $i$.
  If we already balanced for all $j' \le j$,
  then we will skip some values $j'' \ge j$
  and only establish (\ref{eq:machines-balance})
  again for the first $j''$ such that
  \begin{equation*}
      \frac 1 m \sum_i \ell_i(j'') > \frac 1 m \sum_i[\ell_i(j)] + 6 p^2_{\max} .
  \end{equation*}
  It is not hard to see that the balance for job prefixes between $j$ and $j''$ then also holds (with a slightly worse constant).
  To balance the values for $j''$, a careful look at the swapping procedure reveals that it suffices to move jobs, which are scheduled after
  $\frac 1 m \sum_i[\ell_i(j'')] - 3 p^2_{\max}$.
  Such jobs cannot have a due date lower than $d_j$,
  since $\max_i \ell_i(j) \le \frac 1 m \sum_i [\ell_i(j)] + 3 p^2_{\max} \le \frac 1 m \sum_i \ell_i(j'') - 3 p^2_{\max}$. Hence, the swaps do not affect the values $\ell_i(j')$ for $j' \le j$.
  Continuing these swaps, we can establish~\eqref{eq:machines-balance}.
  
  Let us now consider~\eqref{eq:machines-lb}. Suppose that for some job $j$ and machine $i$ we have
  \begin{equation*}
        \ell_i(j) > \min_{j' > j}\left\{ d_{j'} - \frac{1}{m} \sum_{j''=j+1}^{j'} p_{j''} \right\} + \Omega(p_{\max}^2) \ .
  \end{equation*}
  With a sufficiently large hidden constant (compared to~\eqref{eq:machines-balance}) it follows that
  there is a volume of at least $3p_{\max}^2$ of jobs $j'$ with $d_{j'} > d_j$ that are not scheduled
  by this optimal solution. This follows simply from 
  considering the available space on the machines.
  Also with a sufficiently large constant it holds that $\ell_i(j) > 3 p^2_{\max}$.
  In the same way as earlier in the proof, we can find two non-empty sets of jobs $A'$ and $B'$ where
  $A'$ consists only of jobs $j'$ with $d_{j'} > d_j$ that are not scheduled and $B'$ consists only of
  jobs $j'$ with $d_{j'} \le d_j$, which are scheduled on machine $i$. We can swap these two sets and
  retain an optimal solution. Indeed, this may lead to a new violation of~\eqref{eq:machines-balance}.
  In an alternating manner we establish~\eqref{eq:machines-balance} and then perform a swap for~\eqref{eq:machines-lb}.
  Since every time the latter swap is performed the average due dates of the scheduled jobs increases, the process
  must terminate eventually.
\end{proof}
Having this structural property, we solve the problem by dynamic programming:
for every $j = 1,2,\dotsc,n$ and every potential machine-load
pattern $(\ell_1(j),\ell_2(j),\dotsc,\ell_m(j))$ we store the highest volume of jobs in $1,2,\dotsc,j$
that achieves these machine loads or lower machine loads on all machines.
By Lemma~\ref{lem:Pm}, we may restrict ourselves to patterns where machine loads differ pairwise by only $O(p_{\max}^2)$,
but this is not sufficient to obtain our desired running time.
For each job $j$ we will only look at patterns where all machines $i$ satisfy
\begin{equation*}
  \ell_i(j) = \min_{j'>j} \left\{d_{j'} - \frac{1}{m}\sum_{j'>j} p_{j'} \right\} + k, \qquad \text{for } k \in \{-O(p_{\max}^2),\dotsc,O(p_{\max}^2)\}.
\end{equation*}
By the second part of Lemma~\ref{lem:Pm}
it is clear that we can ignore larger values of $k$, but it needs to be clarified why we can
ignore smaller values of $k$ as well. In the case that $\ell_i(j) < \min_{j'>j} \{d_{j'} - \frac{1}{m}\sum_{j'>j} p_{j'} \} - \Omega(p^2_{\max})$
for some $i$ and $j$ (with sufficiently large hidden constants), we may assume with~\eqref{eq:machines-balance} that all machines $i'$ satisfy
$\ell_i(j) \le \min_{j'>j} \{d_{j'} - \frac{1}{m}\sum_{j''=j+1}^{j'} p_{j''} \} - p_{\max}$.
It is not hard to see that starting with such a solution we can greedily add all remaining jobs $j'>j$ to the schedule
without violating any due date. This is still true if we increase the machine loads until we reach $k = -O(p_{\max}^2)$.
It is therefore not necessary to remember any lower load.

For each $j = 1,2,\dotsc,n$, the number of patterns $(\ell_1(j), \ell_2(j),\dotsc,\ell_m(j))$ can now be bounded by $p_{\max}^{O(m)}$, so there are in total $n \cdot p_{\max}^{O(m)}$ states in the dynamic program.
Calculations for each state take $O(m)$ time, because job $j$ is either scheduled as the last job on one of the $m$ machines, or not scheduled at all. Hence, we obtain an $n \cdot p_{\max}^{O(m)}$ running time.

\section{Hardness of ILP with triangular block structure}

In this section we will show that it is NP-hard to decide if there is a feasible solution to an integer linear program of the form

\[
\begin{pmatrix}
A & 0 & \cdots & 0\\
A & A & \ddots & \vdots \\
\vdots &  & \ddots & 0 \\
A & A & \cdots & A 
\end{pmatrix} \cdot \begin{pmatrix}
x_1 \\ x_2 \\ \vdots \\ x_n
\end{pmatrix} \leqslant \begin{pmatrix}
b_1 \\ b_2 \\ \vdots \\ b_m
\end{pmatrix}, \quad \forall_i \ 0 \leq x_i \leq u_i, \quad \forall_i \ x_i \in \mathbb{Z},
\]
even when 
$A$ is a matrix of constant size with integers entries of at most a constant absolute values,
$b_i \in \{0, 1, +\infty\}$ for $i = 1,2,\dotsc,m$, and 
$u_i \in \{0, +\infty\}$ for $i=1,2,\dotsc,n$.
This implies Theorem~\ref{th:trianglefold} by taking $B_1, B_2, \dotsc, B_n$ each as an identity matrix and a negated identity matrix on top of each other, which allows us to easily implement the bounds on the variables.

We will give a reduction from the Subset Sum problem, which asks, given $n$ nonnegative integers $a_1, a_2, \ldots, a_n \in \mathbb{Z}_{\geqslant 0}$, and a target value $t \in \mathbb{Z}_{\geqslant 0}$, whether there exists a subset of $\{a_1, a_2, \ldots, a_n\}$ that sums up to $t$. The Subset Sum problem is weakly NP-hard~\cite{Karp72}, so the reduction will have to to deal with $n$-bit input numbers.

Before we describe the actual reduction, let us start with two general observations that we are going to use multiple times later in the proof.

First, even though the theorem statement speaks, for clarity, about an ILP structure with only ``$\leqslant$'' constraints, we can actually have all three types of constraints, i.e., ``$\leqslant$'', ``$=$'', and ``$\geqslant$''. Indeed, it suffices to append to the matrix $A$ a negated copy of each row, in order to be able to specify for each constraint not only an upper bound but also a lower bound (and use $+\infty$ when only one bound is needed). An equality constraint can then be expressed as an upper bound and a lower bound with the same value.

Second, even though the constraint matrix has a very rigid repetitive structure, we can selectively cancel each individual row, by setting the corresponding constraint to ``$\leq +\infty$'', or each individual column -- by setting the upper bound of the corresponding variable to $0$.

For now, let us consider matrix $A$ composed of four submatrices $B$, $C$, $D$, $E$, arranged in a $2 \times 2$ grid, as follows:
\[
A = \begin{pmatrix} B & C \\ D & E \end{pmatrix}; \quad
\begin{pmatrix}
A &   &   &   & \\
A & A &   &   & \\
A & A & A &   & \\
\vdots & & &\ddots
\end{pmatrix} = \begin{pmatrix}
B & C &   &   &   &   &    \\
D & E &   &   &   &   &    \\
B & C & B & C &   &   &    \\
D & E & D & E &   &   &    \\
B & C & B & C & B & C &    \\
D & E & D & E & D & E &    \\
\vdots & & & & & & \ddots \\
\end{pmatrix}.
\]

In every odd row of $A$'s we will cancel the bottom $(D, E)$ row, and in every even row -- the upper $(B, C)$ row, and similarly for columns, so that we obtain the following structure of the constraint matrix:

\[
\begin{pmatrix}
B & \tikzmark{ta} C & \tikzmark{tb}  &   &   & \tikzmark{tc}  & \tikzmark{td}  &   & \\
\tikzmark{la} D & E &   &   &   &   &   &   & \tikzmark{ra} \\
\tikzmark{lb} B & C & B & C &   &   &   &   & \tikzmark{rb} \\
D & E & D & E &   &   &   &   & \\
B & C & B & C & B & C &   &   & \\
\tikzmark{lc}D & E & D & E & D & E &   &   & \tikzmark{rc} \\
\tikzmark{ld}B & C & B & C & B & C & B & C & \tikzmark{rd} \\
D & E & D & E & D & E & D & E & \\
\vdots & \tikzmark{ba}  & \tikzmark{bb}  &   &   & \tikzmark{bc}  & \tikzmark{bd}  &   & \ddots \\
\end{pmatrix} \cong \begin{pmatrix}
B &   &   &   & \\
D & E &   &   & \\
B & C & B &   & \\
D & E & D & E & \\
\vdots & & & & \ddots
\end{pmatrix}
\]

\begin{tikzpicture}[overlay,remember picture]
\draw ($(pic cs:la)+(0pt,5pt)$) -- ($(pic cs:ra)+(0pt,5pt)$);
\draw ($(pic cs:lb)+(0pt,5pt)$) -- ($(pic cs:rb)+(0pt,5pt)$);
\draw ($(pic cs:lc)+(0pt,5pt)$) -- ($(pic cs:rc)+(0pt,5pt)$);
\draw ($(pic cs:ld)+(0pt,5pt)$) -- ($(pic cs:rd)+(0pt,5pt)$);
\draw ($(pic cs:ta)+(4pt,11pt)$) -- ($(pic cs:ba)+(0pt,0pt)$);
\draw ($(pic cs:tb)+(0pt,11pt)$) -- ($(pic cs:bb)+(0pt,0pt)$);
\draw ($(pic cs:tc)+(0pt,11pt)$) -- ($(pic cs:bc)+(0pt,0pt)$);
\draw ($(pic cs:td)+(0pt,11pt)$) -- ($(pic cs:bd)+(0pt,0pt)$);
\end{tikzpicture}

Now, let us set the four submatrices of $A$ to be $1 \times 1$ matrices with the following values.
\[
B = \begin{pmatrix} 1 \end{pmatrix}, \quad
C = \begin{pmatrix} -2 \end{pmatrix}, \quad
D = \begin{pmatrix} 1 \end{pmatrix}, \quad
E = \begin{pmatrix} -1 \end{pmatrix}.
\]

Let us consider a constraint matrix composed of $2n$ block-rows and $2n$ block-columns.
We set the first constraint to ``$\leq 1$'', and all the remaining constraints to ``$=0$''. Let us denote the variables corresponding to the first column of $A$ by $y_1, y_2, 
\ldots, y_n$, and those to the second column by $z_1, z_2, \ldots, z_n$. We have the following ILP:

\[
\begin{pmatrix}
1 &    &   &    &   &    & \\
1 & -1 &   &    &   &    & \\
1 & -2 & 1 &    &   &    & \\
1 & -1 & 1 & -1 &   &    & \\
1 & -2 & 1 & -2 & 1 &    & \\
1 & -1 & 1 & -1 & 1 & -1 & \\
\vdots & & & & & & \ddots 
\end{pmatrix} \cdot \begin{pmatrix}
y_1 \\ z_1 \\ y_2 \\ z_2 \\ y_3 \\ z_3 \\ \vdots
\end{pmatrix}\quad\begin{matrix}
\leqslant \\ = \\ = \\ = \\ = \\ = \\ \vdots
\end{matrix}\begin{pmatrix}
1 \\ 0 \\ 0 \\ 0 \\ 0 \\ 0 \\ \vdots
\end{pmatrix}
\]

Observe that $z_i = y_i$ and $y_{i+1} = y_1 + \cdots + y_i$ for every $i$.
Since $y_1 \in \{0, 1\}$, it is easy to verify that there are exactly two solutions to this ILP.
Indeed, either $y_i = z_i = 0$ for every $i$, or $y_1 = z_1 = 1$ and $y_{i+1}=z_{i+1} = 2^i$ for every $i$. In other words, either $z = (0,0,0,\ldots)$, or $z = (1, 1, 2, 4, 8, \ldots)$. We will call these two solutions \emph{all-zeros} and \emph{powers-of-two}, respectively.

Now, let us add one more column, namely $(1, 0)$, to matrix $A$, which therefore looks now as follows:
\[ A = \begin{pmatrix} 1 & -2 & 1 \\ 1 & -1 & 0 \end{pmatrix}. \]
The newly added column (and the corresponding variable) shall be cancelled (by setting the corresponding upper bound to $0$) in all but the last copy of $A$, which in turn shall have the other two columns cancelled. Let us call $w$ the variable corresponding to the only non-cancelled copy of the $(1, 0)$ column. Both solutions to the previous ILP extend to the current one, with $w = \sum_i z_i$. Note that, in both solutions, both $\sum_i (y_i - 2 z_i) + w = 0$, and $\sum_i (y_i - z_i) = 0$. Therefore, if we append another copy of the ILP to itself, as follows,
\[
\begin{pmatrix}
\tikzmark{startofmatrix}
1 &    &   &    &   &    & & & & \\
1 & -1 &   &    &   &    & & & & \\
1 & -2 & 1 &    &   &    & & & &\\
1 & -1 & 1 & -1 &   &    & & & &\\
\vdots & & & & \ddots & & & & & \\
1 & -2 & 1 & -2 & \cdots & 1 & & &\\
1 & -1 & 1 & -1 & \cdots & 0 \tikzmark{endoffirst} & & &\\
1 & -2 & 1 & -2 & \cdots & 1 & \tikzmark{startofsecond} 1 & &\\
1 & -1 & 1 & -1 & \cdots & 0 & 1 & -1 &\\
\vdots & & & & & & & & \ddots
\tikzmark{endofmatrix}
\end{pmatrix} \cdot \begin{pmatrix}
\tikzmark{startofvars}
y_1 \\ z_1 \\ y_2 \\ z_2 \\ \vdots \\ w \tikzmark{endoffirstvars} \\ \tikzmark{startofsecondvars} y'_1 \\ z'_1 \\ \vdots
\tikzmark{endofvars}
\end{pmatrix}\quad\begin{matrix}
\tikzmark{startofcons}
\leqslant \\ = \\ = \\ = \\ \vdots \\ = \\ =  \\ \tikzmark{startofsecondcons} \leqslant \\ = \\ \vdots
\end{matrix}\begin{pmatrix}
1 \\ 0 \\ 0 \\ 0 \\ \vdots \\ 0 \\ 0 \tikzmark{endoffirstcons} \\ 1 \\ 0 \\ \vdots
\tikzmark{endofcons}
\end{pmatrix},
\]
the two copies are independent from each other, and we get an ILP that has exactly four feasible solutions: both $z$ and $z'$ can be either all-zeros or powers-of-two, independently, giving four choices in total.
\begin{tikzpicture}[overlay,remember picture]
\draw[dashed] ($(pic cs:startofmatrix)+(-2pt,.9em)$) rectangle ($(pic cs:endoffirst)+(2pt,-2pt)$);
\draw[dashed] ($(pic cs:startofsecond)+(-2pt,.9em)$) rectangle ($(pic cs:endofmatrix)+(2pt,-2pt)$);

\draw[dashed] ($(pic cs:startofvars)+(-10pt,.85em)$) rectangle ($(pic cs:endoffirstvars)+(10pt,-1pt)$);
\draw[dashed] ($(pic cs:startofsecondvars)+(-10pt,.85em)$) rectangle ($(pic cs:endofvars)+(13pt,-1pt)$);

\draw[dashed] ($(pic cs:startofcons)+(-3pt,.85em)$) rectangle ($(pic cs:endoffirstcons)+(10pt,-1pt)$);
\draw[dashed] ($(pic cs:startofsecondcons)+(-3pt,.85em)$) rectangle ($(pic cs:endofcons)+(13pt,-1pt)$);
\end{tikzpicture}

Let $n$ be the number of elements in the Subset Sum instance we reduce from. We copy the above construction $n+1$ times, and we will call each copy a \emph{super-block}. In the last super-block we change the first constraint from ``$\leqslant 1$'' to ``$=1$'', effectively forcing the powers-of-two solution. Therefore, the resulting ILP has exactly $2^n$ feasible solutions -- two choices for each of the first $n$ super-blocks, one choice for the last super-block. We will denote by $z_{i,j}$ the $j$-th $z$-variable in the $i$-th super-block.

Now, we replace the $z$-column of $A$ with three identical copies of it, and each variable $z_{i,j}$ with three variables $p_{i,j}$, $q_{i,j}$, $r_{i,j}$. 

For each $i$, $j$ we will set to $0$ exactly two out of the three upper bounds of $p_{i,j}$, $q_{i,j}$, $r_{i,j}$. Therefore, the solutions of the ILP after the replacement map one-to-one to the solutions of the ILP before the replacement, with $z_{i,j} = p_{i,j} + q_{i,j} + r_{i,j}$. Let $a_1, a_2, \ldots, a_n \in \mathbb{Z}_{\geqslant 0}$ be the elements in the Subset Sum instance we reduce from, and let $t \in \mathbb{Z}_{\geqslant 0}$ be the target value. The upper bounds are set as follows. For every $i$ and for $j=1$, we set $p_{i,1} \leqslant +\infty$ and $q_{i,1}, r_{i,1} \leqslant 0$. For $i = 1, 2, \ldots, n$, we set
\begin{align*}
&p_{i,j} \leqslant + \infty &&\text{and} &&q_{i,j} \leqslant 0 &&\text{if the $(j-1)$-th bit of $a_i$ is zero, and}\\
&p_{i,j} \leqslant 0 &&\text{and} &&q_{i,j} \leqslant +\infty &&\text{if the $(j-1)$-th bit of $a_i$ is one};\end{align*}
in both cases $r_{i,j} \leqslant 0$. For $i = n + 1$ we look at $t$ instead of $a_i$, and we swap the roles of the $q$-variables and $r$-variables, i.e., we set
\begin{align*}
&p_{n+1,j} \leqslant + \infty &&\text{and} &&r_{n+1,j} \leqslant 0 &&\text{if the $(j-1)$-th bit of $t$ is zero, and}\\
&p_{n+1,j} \leqslant 0 &&\text{and} &&r_{n+1,j} \leqslant +\infty &&\text{if the $(j-1)$-th bit of $t$ is one};\end{align*}
and in both cases $q_{n+1,j} \leqslant 0$.

Note that, for $i = 1, 2, \ldots, n$, depending on whether the part of the solution corresponding to the $i$-th super-block is the all-zeros or powers-of-two, $\sum_j q_{i,j}$ equals either $0$ or $a_i$. Hence, the set of the sums of all $q$-variables over all feasible solutions to the ILP is exactly the set of Subset Sums of $\{a_1, a_2, \ldots, a_n\}$. Moreover, $\sum_j r_{n+1,j} = t$.

We need one last step to finish the description of the ILP. We add to matrix A row $(0, 0, 1, -1, 0)$, so it looks as follows:
\[ A = \begin{pmatrix} 1 & -2 & -2 & -2 & 1 \\ 1 & -1 & -1 & -1 & 0 \\ 0 & 0 & 1 & -1 & 0 \end{pmatrix}. \]

The newly added row (and the corresponding constraint) shall be cancelled (by setting the constraint to ``$\leqslant +\infty$'') in all but the last row of the whole constraint matrix. That last constraint in turn shall be set to ``$= 0$'', so that we will have $\sum_i \sum_j q_{i,j} - \sum_i \sum_j r_{i,j} = 0$, i.e., $\sum_i \sum_j q_{i,j} = t$. Hence, the final constructed ILP has a feasible solution if and only if there is a choice of all-zeros and powers-of-two solutions for each super-block that corresponds to a choice of a subset of $\{a_1, a_2, \ldots, a_n\}$ that sums up to $t$. In other words, the ILP has a feasible solution if and only if the Subset Sum instance has a feasible solution.

Finally, let us note that the ILP has $O(n^2)$ variables, and the desired structure.

\bibliographystyle{abbrv}
\bibliography{library}

\end{document}